\documentclass[conference]{IEEEtran}
\usepackage{amsmath,graphicx}
\usepackage{color}
\usepackage{graphicx}
\usepackage{epstopdf}
\usepackage{amsmath}
\usepackage{amssymb}
\usepackage{empheq}

\usepackage{mathrsfs}
\usepackage{hyperref}
\usepackage{tikz}
\usepackage{bm}
\usepackage[english]{babel}
\usepackage{cite}
\usepackage{rotfloat}
\usepackage{mathtools}
\usepackage[font=normalsize,labelfont=bf]{caption}
\usepackage{amsmath}
\usepackage{makecell}
\usepackage{algorithmic}
\usepackage{multirow}
\usepackage{subfigure}
\usepackage{booktabs}
\usepackage{colortbl}
\usepackage{multirow}
\usepackage{hhline}
\usepackage{stfloats}
\usepackage{multicol}
\usepackage{bbm}
\usepackage{cases}
\usepackage[linesnumbered,ruled,lined]{algorithm2e}

\graphicspath{ {Figures/} }
\newcommand{\T}{{\scriptscriptstyle\mathsf{T}}}
\renewcommand{\H}{{\scriptscriptstyle\mathsf{H}}}

\newsavebox{\foobox}

\definecolor{kugray5}{RGB}{224,224,224}

\usepackage[normalem]{ulem}
\newcommand\rsout{\bgroup\markoverwith
	{\textcolor{red}{\rule[0.5ex]{2pt}{0.8pt}}}\ULon}



\makeatletter
\newcommand{\ALOOP}[1]{\ALC@it\algorithmicloop\ #1%
	\begin{ALC@loop}}
	\newcommand{\ENDALOOP}{\end{ALC@loop}\ALC@it\algorithmicendloop}

\makeatother

\usepackage{etoolbox}
\let\mybibitem\bibitem
\renewcommand{\bibitem}[1]{%
	\ifstrequal{#1}{nature}
	{\color{blue}\mybibitem{#1}}
	{\color{black}\mybibitem{#1}}%
}

\graphicspath{ {Figures/} }

\newtheorem{theorem}{\textbf{Theorem}}
\newtheorem{remark}{\textbf{Remark}}

\newtheorem{proof}{Proof}

\newcommand{\epr}{\hfill\(\Box\)}

\DeclareCaptionLabelSeparator{periodspace}{.\quad}

\addto\captionsenglish{}
\newcommand\nbthis{\addtocounter{equation}{1}\tag{\theequation}}

\newcommand{\norm}[1]{\left\lVert#1\right\rVert} 
\newcommand{\normshort}[1]{\lVert#1\rVert} 
\newcommand{\abs}[1]{\left|#1\right|} 
\newcommand{\absshort}[1]{|#1|} 

\newcommand{\tr}[1]{\text{trace}\left(#1\right)} 
\newcommand{\trshort}[1]{\text{trace}(#1)} 

\newcommand{\diag}[1]{\mathtt{diag}\left\{#1\right\}} 

\newcommand{\re}[1]{\mathfrak{R}{\left(#1\right)}}
\newcommand{\im}[1]{\mathfrak{I}{\left(#1\right)}}
\allowdisplaybreaks

\newcommand{\mean}[1]{\mathbb{E} \left\{#1\right\}}
\newcommand{\meanshort}[1]{\mathbb{E} \{#1\}}

\usepackage{setspace}

\newcommand{\mR}{{\mathbf{R}}}
\newcommand{\mH}{{\mathbf{H}}} 

\newcommand{\mA}{{\mathbf{A}}}
\newcommand{\mS}{{\mathbf{S}}}

\newcommand{\mI}{\textbf{\textbf{I}}}

\newcommand{\mB}{{\mathbf{B}}}

\newcommand{\mD}{{\mathbf{D}}}
\newcommand{\mX}{{\mathbf{X}}}
\newcommand{\mY}{{\mathbf{Y}}}
\newcommand{\mG}{{\mathbf{G}}}
\newcommand{\mF}{{\mathbf{F}}}

\newcommand{\mZ}{{\mathbf{Z}}}

\newcommand{\mN}{{\mathbf{N}}}

\newcommand{\mJ}{{\mathbf{J}}}

\newcommand{\setC}{\mathbb{C}} 
\newcommand{\setR}{\mathbb{R}}

\newcommand{\ve}{{\mathbf{e}}} 
\newcommand{\vs}{{\mathbf{s}}}
\newcommand{\vx}{{\mathbf{x}}}
\newcommand{\vy}{{\mathbf{y}}}

\newcommand{\vv}{{\mathbf{v}}}
\newcommand{\vn}{{\mathbf{n}}}

\newcommand{\vz}{{\mathbf{z}}} 
\newcommand{\vh}{{\mathbf{h}}}

\newcommand{\vb}{{\mathbf{b}}}
\newcommand{\vw}{{\mathbf{w}}}
\newcommand{\va}{{\mathbf{a}}}

\newcommand{\vf}{\mathbf{f}}

\def\b0{{\pmb{0}}}

\newcommand{\Nr}{N_\mathtt{r}}
\newcommand{\Nt}{N_\mathtt{t}}

\newcommand{\sigmac}{\sigma_{\mathtt{c}}^2}
\newcommand{\sigmas}{\sigma_{\mathtt{s}}^2}

\newcommand{\noise}{\sigma^2}
\newcommand{\Pt}{P_{\mathtt{t}}}

\begin{document}
	\title{Multi-Static Cell-Free Massive MIMO ISAC: Performance Analysis and Power Allocation}
	\author{\IEEEauthorblockN{Nhan~Thanh~Nguyen\IEEEauthorrefmark{1}, Tianyu~Fang\IEEEauthorrefmark{1}, Hien~Quoc~Ngo\IEEEauthorrefmark{2}, and Markku Juntti\IEEEauthorrefmark{1}}
		\IEEEauthorblockA{\IEEEauthorrefmark{1}Centre for Wireless Communications, University of Oulu, P.O.Box 4500, FI-90014, Finland}
		\IEEEauthorblockA{\IEEEauthorrefmark{2}School of Electronics, Electrical Engineering and Computer Science, Queen’s University Belfast, UK}
		Emails: \{nhan.nguyen, tianyu.fang, markku.juntti\}@oulu.fi; hien.ngo@qub.ac.uk}
	\maketitle
	
	\begin{abstract}
		In this work, we consider a cell-free massive multiple-input multiple-output (MIMO) integarted sensing and communications (ISAC) system with maximum-ratio transmission schemes combined with multistatic radar-type sensing. Our focus lies on deriving closed-form expressions for the achievable communications rate and the Cramér-Rao lower bound (CRLB), which serve as performance metrics for communications and sensing operations, respectively. The expressions enable us to investigate important operational characteristics of multistatic cell-free massive MIMO-ISAC, including the mutual effects of communications and sensing as well as the advantages stemming from using numerous distributed antenna arrays for each functionality. Furthermore, we optimize the power allocation among the access points to maximize the communications rate while guaranteeing the CRLB constraints and total transmit power budget. Extensive numerical results are presented to validate our theoretical findings and demonstrate the efficiency of the proposed power allocation approach.
	\end{abstract}
	
	\begin{IEEEkeywords}
		Multi-static, integrated sensing and communications (ISAC), massive MIMO, maximum-ratio transmission.
	\end{IEEEkeywords}
	\IEEEpeerreviewmaketitle
	
	\section{Introduction}
	Dual-functional beamforming and waveform designs have been a recent focus in the literature, aiming to enhance resource efficiency in integrated sensing and communications (ISAC) systems \cite{liu2018mu}. While existing designs  mainly investigated mono-static ISAC setups \cite{liu2020joint, johnston2022mimo, liu2022joint, nguyen2023joint, krishnananthalingam2024constant, nguyen2024massive, hatami2024waveform, nguyen2023multiuser, nguyen2024joint}, recent studies reveal that cell-free massive MIMO (CFmMIMO) holds promise for enabling multi-static ISAC operations with improved communications and sensing performance.
	
    In \cite{elfiatoure2024multiple}, Elfiatoure \textit{et al.} derived closed-form expressions for the communication spectral efficiency (SE) and the main lobe-to-average-side lobe ratio (MASR) in the sensing zones of a CFmMIMO ISAC system. They further introduced a joint mode selection and power control design to maximize SE fairness while maintaining MASR requirements.  Rivetti \textit{et al.} \cite{rivetti2024secure} addressed security concerns in a CFmMIMO ISAC network, where the target acts as an eavesdropper, by proposing a centralized ISAC waveform that introduces artificial noise to disrupt the eavesdropper’s channel. The optimal precoding and artificial noise covariance matrices were derived via a semi-definite relaxation approach to minimize a constrained Cramer-Rao lower bound (CRLB). Behdad \textit{et al.} \cite{behdad2024multi} developed a maximum a posteriori ratio test detector for target detection and a power allocation algorithm to maximize the sensing signal-to-interference-plus-noise ratio (SINR) while ensuring minimum communications SINR.  Similar CF mMIMO ISAC designs were considered in \cite{demirhan2023cell} and \cite{demirhan2023cell_asilomar} but with the focus on maximizing sensing signal-to-noise ratios (SNRs). In \cite{mao2024communication}, the beamforming designs was proposed to maximized the communications sum rate subject to the sensing beampattern accuracy. Graph neural networks are leveraged for CF mMIMO ISAC beamforming designs in \cite{demirhan2024learning} and \cite{adhikary2024holographic}. Whereas, Zeng \textit{et al.} \cite{zeng2024multi} explored a multi-static ISAC approach utilizing network-assisted full-duplex CF networks, developing a deep Q-network-based power control method to balance communications and sensing in the uplink and downlink of access points.
	

	In this work, we consider a CF mMIMO-ISAC system, wherein APs jointly serve multiple user equipments (UEs) and a sensing target. The MRT method is employed for precoding at distributed APs, while the transmit power associated with the APs are optimized at the CPU. To investigate the operational characteristic of multistatic CF mMIMO-ISAC, we first deriving closed-form expressions for the achievable communications rate and the CRLB, which serve as performance metrics for communications and sensing operations, respectively. Based on that, we provide useful insights on the impact of sensing on communications as well as the benefits from having multiple distributed APs. Furthermore, we optimize the power allocation among the APs to maximize the communications rate subject to the constraints on the minimum CRLB for target's angle estimations and total transmit power budget. Our extensive numerical results validate our theoretical analyses and demonstrate the efficiency of the proposed power allocation approach. 
	
	\section{Signal Model}
	\label{sec_system_nodel}

	We consider a multi-static CF massive MIMO ISAC system, including $L$ APs, $K$ single-antenna communications UEs, and a sensed target. Each AP is equipped with $\Nt$ transmit and  $\Nr$ receive antennas. The APs transmit probing signals to the target at a given angle of interest and data signals to the UEs at the same time and frequency.
	
	\subsection{Communications Model}
	
	Denote the transmit vector at time slot $t$ by $\vs_{t} = [s_{1t}, \ldots, s_{Kt}] \in {\mathbb{C}}^{K \times 1}$, where $\mean{\vs_{t} \vs_{t}^{\H}}=\mI_K$ and $s_{k t}$ is the signal intended for the $k$-th UE. Furthermore, let $\mS = [\vs_{1}, \ldots, \vs_{T}] \in \setC^{K \times T}$, where $T \gg 1$ is the length of the radar/communications frame. The data streams are assumed to be independent of each other such that $\mS \mS^\H \approx T \mI_K$~\cite{liu2021cramer}. Let $\mF_{\ell} \in \setC^{\Nt \times K}$ denote the precoding matrix at AP $\ell$. Then, the $\Nt \times 1$ transmit signal vector during time slot $t$ is given as $\vx_{\ell  t} = \mF_{\ell} \vs_{t} = \sum_{k=1}^{K} \vf_{k \ell} s_{k t}$, where $\vf_{k \ell}$ is the $k$th column of $\mF_{\ell}$. The overall dual-functional transmit waveform is given by $\mX_{\ell} \triangleq \mF_{\ell} \mS = [\vx_{\ell 1}, \ldots, \vx_{\ell T}] \in \setC^{\Nt \times T}$.
	
	At time slot $t$, the received signal at UE $k$ is given as
	\begin{align*}
		y_{kt} \!=\! \sum\nolimits_{\ell =1}^{L} \vh_{k \ell}^\H \vf_{k \ell} s_{k t} \!+\!  \sum\nolimits_{\ell=1}^{L} \vh_{k \ell}^\H \sum\nolimits_{j \neq k} \vf_{j \ell} s_{j t} \!+\! n_{kt}. \nbthis \label{eq_y_kl}
	\end{align*}
	where $n_{kt}$ is additive white Gaussian noise (AWGN) following the distribution $\mathcal{CN}(0, \sigmac)$, and $\vh_{k \ell}$ denote the channel between AP $\ell$ and UE $k$, modeled as $\vh_{k \ell} = \beta_{k \ell}^{1/2} \bar{\vh}_{k \ell}$ \cite{mollen2016uplink}, with $\beta_{k \ell}$ and $\bar{\vh}_{k \ell} \sim \mathcal{CN}(0, \mI_{\Nt})$ being the large-scale coefficient and small-scale Rayleigh fading channels, respectively. We assume the  minimum mean square error (MMSE) estimator for channel estimation. Let $\hat{\vh}_{k \ell}$ and $\ve_{k\ell}$ respectively denote the estimate and estimation error of $\vh_{k \ell}$. We have $\hat{\vh}_{k \ell} \sim \mathcal{CN}(\bm{0}, \xi_{k \ell} \mI_{\Nt})$ and $\ve_{k \ell} \sim \mathcal{CN}(\bm{0}, \epsilon_{k \ell} \mI_{\Nt})$, where \cite{ngo2017total}
	\begin{align*}
		\xi_{k \ell} &= \frac{\tau_{\mathtt{p}} p_{\mathtt{p}} \beta_{k \ell}^2}{\tau_{\mathtt{p}} p_{\mathtt{p}} \sum_{j=1}^{K} \beta_j  \absshort{\bm{\psi}_{j \ell}^\H \bm{\psi}_{k \ell}}^2 + \sigmac}, \nbthis \label{eq_xi_k}\\
		\epsilon_{k \ell} = \beta_{k \ell} - \xi_{k \ell} &= \frac{\beta_{k \ell} (\tau_{\mathtt{p}} p_{\mathtt{p}} \sum_{j\neq k} \beta_{j \ell}  \absshort{\bm{\psi}_{j \ell}^\H \bm{\psi}_{k \ell}}^2 + \sigmac )}{\tau_{\mathtt{p}} p_{\mathtt{p}} \sum_{j=1}^{K} \beta_j  \absshort{\bm{\psi}_{j \ell}^\H \bm{\psi}_{k \ell}}^2 + \sigmac}, \nbthis \label{eq_epsilon_k}
	\end{align*} 
	with $\bm{\psi}_{k \ell} \in \setC^{\tau_{\mathtt{p}} \times 1}$  $(\|\bm{\psi}_{k \ell}\|^2=1)$ being the pilot sequence transmitted by UE $k$ to AP $\ell$, $\tau_\mathtt{p}$ is the length of the pilot sequences, and $p_{\mathtt{p}}$ is the average power of the training symbols.
	
	Let $\gamma_{k \ell}$ and $\eta_{k \ell}$ denote the power factor allocated for communications and sensing within the $k$th data stream at AP $\ell$. Denote $\bm{\Gamma}_{\ell} \triangleq \diag{\sqrt{\gamma_{1 \ell}}, \ldots, \sqrt{\gamma_{K \ell}}} \in  \setC^{K \times K}$ and $\bar{\bm{\eta}}_{\ell} \triangleq [\sqrt{\eta_{1 \ell}},\ldots,\sqrt{\eta_{K \ell}}]^\T \in \setC^{K \times 1}$. Furthermore, let $\vv_{\ell} \in \mathbb{C}^{\Nt \times 1}$ be the precoding vector for sensing, and let $\hat{\mH}_{\ell} = [\hat{\vh}_{1 \ell}, \ldots, \hat{\vh}_{K \ell}]$. With the MRT precoder for communications, the precoding matrix is given as $\mF_{\ell} = \hat{\mH}_{\ell} \bm{\Gamma}_{\ell} + \vv_{\ell} \bar{\bm{\eta}}_{\ell}^\T$.

	\subsection{Radar Model}
	Each AP receives the echo signal of its own transmitted signal as well as those from the other APs. Thus, the received discrete-time radar sensing signal at AP $p$ is given as
	\begin{align*}
		\mY^{\mathtt{s}}_{p} &= \sum\nolimits_{\ell =1}^{L} \alpha_{\ell p} \mG_{\ell p}(\theta_{p}, \phi_{\ell}) \mX_{\ell} + \mN^{\mathtt{s}}_{p}, \nbthis \label{eq_radar_nodel}
	\end{align*}
	where $\alpha_{\ell p}$ is the gain of the channel from the transmit AP $\ell$ to receive AP $p$, $\mN^{\mathtt{s}}_{p}$ is an AWGN matrix with entries distributed as $\mathcal{CN}(0, \sigmas)$. Furthermore, $\mG_{\ell p}(\theta_{p}, \phi_{\ell}) \in \setC^{\Nr \times \Nt}$ is the two-way channel in the desired sensing directions, modeled as $\mG_{\ell p}(\theta_{p}, \phi_{\ell}) = \vb_{p}(\theta_{p}) \va_{\ell}^\H(\phi_{\ell})$ \cite{johnston2022mimo, liu2020joint}, 
	where $\theta_{p} (\phi_{\ell}) \in \left[-\frac{\pi}{2}, \frac{\pi}{2}\right]$ are the angle of arrival (departure) (AoA/AoD) to (from)  AP $p (\ell)$. Furthermore, $\va_{\ell}(\phi_{\ell})$ and $\vb_{p}(\theta_{p})$ are the transmit and receive steering vectors, respectively. In the following analysis, we drop $\theta_{p}$ and $\phi_{\ell}$ for ease of exposition. We assume that the APs employ uniform linear arrays (ULA) with half-wavelength antenna spacing. The steering vector $\vb_{p}$ is modeled as $\vb_{p} = \left[ e^{-j \pi \sin(\theta_{p})}, e^{-j2\pi \sin(\theta_{p})}, \ldots, e^{j\Nr \pi \sin(\theta_{p})} \right]^\T$ \cite{liu2021cramer}, and  $\va_{\ell}$ is modeled similarly.

	From \eqref{eq_radar_nodel} and $\mX_{\ell} = \mF_{\ell} \mS$ with $\mF_{\ell} = \hat{\mH}_{\ell} \bm{\Gamma}_{\ell} + \vv_{\ell} \bar{\bm{\eta}}_{\ell}^\T$, we have $\mY^{\mathtt{s}}_{p}   = \alpha_{\ell p} \vb_{p} \va_{\ell}^\H \vv_{\ell} \bar{\bm{\eta}}_{\ell}^\T  \mS + \tilde{\mN}^{\mathtt{s}}_{p}$, 
	where $\tilde{\mN}^{\mathtt{s}}_{p} \triangleq \alpha_{\ell p} \vb_{p} \va_{\ell}^\H \hat{\mH}_{\ell} \bm{\Gamma}_{\ell} \mS + \vn^{\mathtt{s}}_{p}$. It is observed that for an arbitrary $\hat{\mH}_{p} \bm{\Gamma}_{p}$, the sensing beamformer $\vv_{\ell} = \va_{\ell}$ maximizes the received echo signal power at the BS. Thus, $\mF_{\ell}$ can be given as
	\begin{align*}
		\mF_{\ell} = \hat{\mH}_{\ell} \bm{\Gamma}_{\ell} + \va_{\ell} \bar{\bm{\eta}}_{\ell}^\T. \nbthis \label{eq_F}
	\end{align*}
	
	\section{Communications and Sensing Performance}
	\label{sec_perf_analysis}
	In this section, we derive the achievable rate and the CRLB to evaluate the communications and sensing performance. 
	
	\subsection{Communications Performance}
	We rewrite \eqref{eq_y_kl} as
	\begin{align*}
		y_{kt} &= \mathtt{DS}_k s_{kt} + \mathtt{BU}_{k} s_{kt} + \sum\nolimits_{j \neq k} \mathtt{UI}_{kj} s_{jt} + n_{kt}, \nbthis \label{eq_ykl}
	\end{align*}
	where $\mathtt{DS}_k \triangleq \mean{ \sum_{\ell =1}^{L} \vh_{k \ell}^\H \vf_{k \ell}}$, $ \mathtt{BU}_{k} \triangleq \sum_{\ell =1}^{L}  \vh_{k \ell}^\H \vf_{k \ell} - \mean{ \sum_{\ell =1}^{L}  \vh_{k \ell}^\H \vf_{k \ell}}$, and $\mathtt{UI}_{kj} \triangleq \sum_{\ell =1}^{L}  \vh_{k \ell}^\H  \vf_{j \ell}$ represent expectations associated with the desired signal, beamforming uncertainty, and inter-UE interference, respectively. From \eqref{eq_ykl}, the achievable rate of the $k$-th UE is given by $R_{k} = \bar{\tau} \log_2 \left(1 \! +\! {\abs{\mathtt{DS}_k}^2}/{\left(\mean{\abs{\mathtt{BU}_{k}}^2} \!+\!  \sum\nolimits_{j \neq k} \mean{\abs{\mathtt{UI}_{kj}}^2} \!+\! \sigmac\right)} \right)$,
	where $\bar{\tau} \triangleq \left(\tau_{\mathtt{c}} - \tau_{\mathtt{p}}\right)/\tau_{\mathtt{c}}$. Its closed-form expression is presented in the following theorem.
	\begin{theorem}
		\label{theo_SE}
		The achievable rate for UE $k$ is given by
			\begin{align*}
				R_k(\bm{\Psi}) \! = \!\bar{\tau} \log_2 \!\left(\!\! 1\! +\! \frac{\Nt^2 \left(\bm{\xi}_k^\T \bar{\bm{\gamma}}_k\right)^2} {\Nt \sum\nolimits_{j=1}^K \left( \hat{\bm{\beta}}_{kj}^\T \bm{\gamma}_j \!+\! \bm{\beta}_{k}^\T \bm{\eta}_j \!\!\right) \!+\! \sigmac}\!\right)\!, \nbthis \label{eq_SE_theo}
			\end{align*}
			where $\bm{\xi}_k = [\xi_{k1}, \ldots, \xi_{kL}]^\T$, $\bar{\bm{\gamma}}_{k} = [\sqrt{\gamma_{k1}},\ldots,\sqrt{\gamma_{kL}}]^\T$,  $\hat{\bm{\beta}}_{kj} = [\beta_{k1} \xi_{j1}, \ldots, \beta_{kL} \xi_{jL}]^\T$, $\bm{\beta}_{k} = [\beta_{k1}, \ldots, \beta_{kL}]^\T$, and $\bm{\Psi} = \{ \{\gamma_{k \ell}\}_{k=1,\ell=1}^{K,L}, \{\eta_{k \ell}\}_{k=1,\ell=1}^{K,L} \}$.
	\end{theorem}

	\begin{proof}
		See Appendix \ref{appd_SE}. \epr
	\end{proof}
	
	
	
	\begin{remark}
		\label{rm_effect_of_sensing}
		It is observed from \eqref{eq_SE_theo} that, due to the term $\Nt \sum_{j=1}^K \bm{\beta}_{k}^\T \bm{\eta}_j$ in the denominator of the SINR, sensing introduces additional beamforming uncertainty and inter-UE interference, and thus causes performance degradation to the communications subsystem. 
		However, note that the numerator of the SINR term in increases with $L^2 \Nt^2$, while the denominator only increases with $L \Nt$. Therefore, the CF mMIMO ISAC system becomes interference-free as $L \Nt \rightarrow \infty$, similar to conventional CF mMIMO systems without any sensing function. In other words, deploying a very large number of transmit antennas and/or APs can mitigate the impact of sensing on communications performance.
	\end{remark}

	\subsection{Sensing Performance}
	
	For the sensing function, we are interested in characterizing the CRLB for estimating the target angles. Note here that AoD at an AP is equal to the AoA of the echo signals from all APs. This implies that estimating $\{\theta_{p}\}_{p =1}^L$ is sufficient for all the AoAs and AoDs for the APs. Therefore, for low-complexity sensing operation, we assume that each AP estimates its own AoA, i.e., AP $p$ estimates $\theta_{p}$.  To derive the CRLB, we first rewrite \eqref{eq_radar_nodel} as
	\begin{align*}
		\mY^{\mathtt{s}} _{p} = \mX^{\mathtt{s}}_{p} + \mN^{\mathtt{s}}_{p} , \nbthis \label{eq_radar_model_1}
	\end{align*}
	where 
	\begin{align*}
		&\mX^{\mathtt{s}}_{p} \triangleq \sum_{\ell =1}^{L} \alpha_{\ell p}  \vb_{p}  \bar{\va}_{\ell}^\H \mD_{\ell} \mX = \sum_{\ell = 1}^{L} \alpha_{\ell p} \mB_{\ell p} \mX \in \setC^{\Nr \times T}, \nbthis \label{def_Xlp} \\
		&\mX \triangleq [\mX_1; \ldots; \mX_L]  \in \setC^{L\Nt \times T}, \nbthis \label{def_X} \\
		&\mB_{\ell p} \triangleq \vb_{p}  \bar{\va}_{\ell}^\H \mD_{\ell}  \in \setC^{\Nr \times L\Nt}, \nbthis \label{def_Blplp}\\
		&\bar{\va}_{\ell}^\H \triangleq \left[\bm{0}, \ldots, \va_{\ell}^\H, \ldots, \bm{0}\right]  \in \setC^{1 \times L \Nt},  \nbthis \label{def_abarl}\\
		&\mD_{\ell} \triangleq \mathtt{blkdiag} \{ \bm{0}_{\Nt}, \ldots, \mI_{\Nt}, \ldots,   \bm{0}_{\Nt}\} \in \setR^{L\Nt \times L\Nt}.  \nbthis \label{def_Dl}
	\end{align*}
	In \eqref{def_X}, we stack matrices $\{\mX_\ell\}_{\ell=1}^L$; in \eqref{def_abarl}, $\bar{\va}_{\ell}^\H$ includes non-zeros elements with indices in the range $\mathcal{R}_{\ell} \triangleq [(\ell - 1)\Nt, \ell \Nt]$; in \eqref{def_Dl}, $\mD_{\ell}$ is a block-diagonal matrix, and only the $\ell$th block with entries on the rows and columns specified by $\mathcal{R}_{\ell}$ are nonzero.
	Let $\tilde{\vz}^{\mathtt{s}} _{p} = \mathtt{vec}(\mZ^{\mathtt{s}} _{p})$ for $\vz \in \{\vx, \vy, \vn\}$ and $\mZ \in \{\mX, \mY, \mN\}$. Then,   \eqref{eq_radar_model_1} can be vectorized as $\tilde{\vy}^{\mathtt{s}} _{p} = \tilde{\vx}^{\mathtt{s}}_{p} + \tilde{\vn}^{\mathtt{s}}_{p}$. 
	
	At AP $p$, the parameters to estimate are $\bm{\omega}_{p} \triangleq [\theta_{p}, \tilde{\bm{\alpha}}_{p p}]^\T$, where $\tilde{\bm{\alpha}}_{p p} = [\re{{\alpha}_{p p}}, \im{{\alpha}_{p p}}]$.  The Fisher information matrix (FIM) for estimating $\bm{\omega}_{p}$ is given by~\cite{liu2021cramer, song2023intelligent, bekkerman2006target}
	\begin{align*}
		\mJ_{\bm{\omega}_{p}} =
		\begin{bmatrix}
			J_{\theta \theta}^{(p)} & \mJ_{\theta \tilde{\bm{\alpha}}}^{(p)} \\
			\left(\mJ_{\theta \tilde{\bm{\alpha}}}^{(p)}\right)^\T & \mJ_{\tilde{\bm{\alpha}} \tilde{\bm{\alpha}}}^{(p)} \\
		\end{bmatrix}, \nbthis \label{eq_FIM}
	\end{align*}
	where $J_{\theta \theta}^{(p)} = \frac{2}{\sigmas}\re{\frac{\partial (\tilde{\vx}^{\mathtt{s}}_{p})^\H}{\partial \theta_{p}} \frac{\partial \tilde{\vx}^{\mathtt{s}}_{p}}{\partial \theta_{p}}}$, $\mJ_{\theta \tilde{\bm{\alpha}}}^{(p)} = \frac{2}{\sigmas} \re{\frac{\partial (\tilde{\vx}^{\mathtt{s}}_{p})^\H}{\partial \theta_{p}} \frac{\partial \tilde{\vx}^{\mathtt{s}}_{p}}{\partial \tilde{\bm{\alpha}}}}$, and $\mJ_{\tilde{\bm{\alpha}} \tilde{\bm{\alpha}}}^{(p)} = \frac{2}{\sigmas} \re{\frac{\partial (\tilde{\vx}^{\mathtt{s}}_{p})^\H}{\partial \tilde{\bm{\alpha}}} \frac{\partial \tilde{\vx}^{\mathtt{s}}_{p}}{\partial \tilde{\bm{\alpha}}}}$.
	The CRLB for $\theta_{p}$ is derived in the following theorem.
	\begin{theorem}\label{theo_CRB}
		The CRLB for $\theta_{p}$ at the $p$th AP is given as
		\begin{align*}
			\mathtt{CRLB}_{p}^{\theta}(\bm{\Psi})  = \frac{\bar{\sigma}_{\mathtt{s}} \tau_{p}(\bm{\Psi})}{ \tilde{\tau}_{p}(\bm{\Psi}) \tau_{p} (\bm{\Psi}) - \abs{ \hat{\tau}_{p} (\bm{\Psi})}^2}, \nbthis \label{eq_CRLB_theo}
		\end{align*}
		where
		\begin{align*}
			&\tau_{p}(\bm{\Psi}) \triangleq c_{p,1} \bm{\xi}_{p}^\T \bm{\gamma}_{p} + c_{p,2} \norm{\bar{\bm{\eta}}_{p}}^2, \nbthis \label{def_tau} \\
			&\tilde{\tau}_{p}(\bm{\Psi}) \triangleq \tilde{c}_{p,1} \bm{\xi}_{p}^\T \bm{\gamma}_{p} + \tilde{c}_{p,2} \norm{\bar{\bm{\eta}}_{p}}^2,  \nbthis \label{def_tau_tilde} \\
			&\hat{\tau}_{p}(\bm{\Psi}) \triangleq \hat{c}_{p,1} \bm{\xi}_{p}^\T \bm{\gamma}_{p} + \hat{c}_{p,2} \norm{\bar{\bm{\eta}}_{p}}^2,  \nbthis \label{def_tau_hat}
		\end{align*}
		In   \eqref{def_tau}, $c_{p,1} \triangleq \trshort{ [ \mB_{p p}^\H \mB_{p p} ]_{pp}}$ and $c_{p,2} \triangleq \trshort{ \tilde{\mA}_{pp} [ \mB_{p p}^\H \mB_{p p} ]_{pp}}$; in \eqref{def_tau_tilde}, $\tilde{c}_{p,1} \triangleq \trshort{ [ \dot{\tilde{\mB}}_{p}^\H \dot{\tilde{\mB}}_{p} ]_{pp}}$ and  $\tilde{c}_{p,2} \triangleq \trshort{ \tilde{\mA}_{pp} [\dot{\tilde{\mB}}_{p}^\H \dot{\tilde{\mB}}_{p} ]_{pp}}$; in \eqref{def_tau_hat},  $\hat{c}_{p,1} \triangleq \trshort{ [\mB_{p p}^\H \dot{\tilde{\mB}}_{p} ]_{pp}}$ and  $\hat{c}_{p,2} \triangleq \trshort{ \tilde{\mA}_{pp} [\mB_{p p}^\H \dot{\tilde{\mB}}_{p} ]_{pp}}$. Here, $\dot{\tilde{\mB}}_{p} \triangleq \sum_{\ell=1}^{L} \alpha_{\ell p} (\dot{\vb}_{p}  \bar{\va}_{\ell}^\H + \delta_{\ell p} \vb_{p}  \dot{\bar{\va}}_{\ell}^\H ) \mD_{p}$ with  $[\mZ]_{p p}$ being the sub-matrix containing entries on rows and columns $\{(p-1)\Nt + 1, \ldots,  p \Nt\}$ in matrix $\mZ$. Furthermore, we have $\bm{\gamma}_p \triangleq [\gamma_{1 p}, \ldots, \gamma_{K p}]^\T$, and $\tilde{\mA}_{p p} \triangleq \va_p  \va_{p}^\H$.
	\end{theorem}

 \begin{proof}
     See Appendix \ref{appd_CRB}. \epr
 \end{proof}

	\section{Power Allocation} 
	With the derived achievable rate and CRLB based on MRT precoding, we propose a power allocation scheme to maximizes the communications sum rate while ensuring the sensing CRB and power constraints. 
	The problem is formulated as
	\begin{subequations}\label{problem1}
		\begin{align*} 
			\underset{\substack{\bm{\Psi}}}{\textrm{maximize}} \quad &  \sum\nolimits_{\substack{k=1}}^K R_k(\bm{\Psi})   \nbthis \label{obj_rate} \\
			\textrm{subject to} \quad
			&\mathtt{CRLB}_{\theta_p}(\bm{\Psi})\leq \mathtt{CRLB}_{\theta_p}^\mathtt{th},\ \forall p, \nbthis \label{cons_CRB} \\
			&\Nt \sum\nolimits_{\ell =1}^{L} \bm{\xi}_{\ell}^\T \bm{\gamma}_{\ell} + \norm{\bar{\bm{\eta}}_\ell}^2 \leq P_{\mathtt{t}}, \nbthis \label{cons_power}
		\end{align*}
	\end{subequations}
	where $\mathtt{CRLB}_{\theta_p}^\mathtt{th}$ is a threshold to guarantee the sensing performance, and $P_{\mathtt{t}}$ is the power budget for all APs. The left-hand-side of \eqref{cons_power} is the average total transmit power at all the APs, i.e., $\sum_{\ell =1}^{L} \mean{\tr{\mF_{\ell} \mF_{\ell}^\H}}$, whose detailed derivations are omitted here due to the space constraint. While \eqref{cons_power} is convex,  objective function \eqref{obj_rate} and constraint \eqref{cons_CRB} are not. We overcome this by leverageing SCA, as elaborated next.
	
\subsection{Proposed Solution}

	\subsubsection{Address Objective Function \eqref{obj_rate}}
	We first introduce variable $\omega_{k\ell} \triangleq \sqrt{\gamma_{k \ell}}, \forall k, \ell$ and denote $\bm{\omega}_{k} \triangleq [\omega_{k1}, \ldots, \omega_{kL}]$. The objective function can be rewritten as $R_k(\bm{\omega}_k,\bm\Psi)=\bar \tau \log_2 \!\left(\!\! 1\! +\! \frac{\Nt^2 (\bm{\xi}_k^\T \bm{\omega}_{k})^2}{\Nt \sum_{j=1}^K \left( \hat{\bm{\beta}}_{kj}^\T {\bm{\gamma}}_j \!+ \bm{\beta}_{k}^\T {\bm{\eta}}_j \!\right) \!+\! \sigmac}\!\right)$.
To tackle the non-convexity of $R_k(\bm{\omega}_k,\bm\Psi)$, we introduce the following concave lower bound of $R_k(\bm{\omega}_k,\bm\Psi)$ \cite{fang2024beamforming}:
 \begin{align*}
    & f(\bm{\omega}_k,\bm\Psi)= \log_2(1+\zeta_k^{(i)})-\zeta_k^{(i)}+2\Nt\iota_k^{(i)}\sqrt{1+\zeta_k^{(i)}}\bm\xi_k^\T\bm{\omega}_k \notag\\
     &-\abs{\iota_k^{(i)}}^2\left(\Nt^2\abs{\bm{\xi}_k^\T \bm{\omega}_{k} }^2+\Nt \sum\nolimits_{j=1}^K \left(  \hat{\bm{\beta}}_{kj}^\T {\bm{\gamma}}_j \!+ \bm{\beta}_{k}^\T {\bm{\eta}}_j \right) \!+\! \sigmac\right),
 \end{align*}
 where $\xi_k^{(i)}$ and $\iota_k^{(i)}$ are constants in iteration $i$, given as
 \begin{align*}
     \zeta_k^{(i)} &= \frac{\Nt^2 \abs{\bm{\xi}_k^\T \bm{\omega}_{k}^{(i)} }^2}{\Nt \sum_{j=1}^K \left( \hat{\bm{\beta}}_{kj}^\T \bm{\gamma}_j^{(i)}  \!+ \bm{\beta}_{k}^\T  {\bm{\eta}}_j^{(i)} \right) \!+\! \sigmac}, \nbthis \label{eq_zeta} \\
      \!\iota_k^{(i)} \! &= \! \frac{\Nt\sqrt{1\!+\!\zeta_k^{(i)}}\bm\xi_k^\T\bm{\omega}_k^{(i)}}{\Nt^2 \abs{\bm{\xi}_k^\T \bm{\omega}_{k}^{(i)} }^2\!\!\!+\!\Nt \sum\limits_{j=1}^K \left( \hat{\bm{\beta}}_{kj}^\T \bm{\gamma}_j^{(i)} \!+ \bm{\beta}_{k}^\T {\bm{\eta}}_j^{(i)} \right) \!+\! \sigmac}. \nbthis \label{eq_iota} 
 \end{align*}
 Note that $\gamma_{kl}=\omega_{kl}^2$ is still non-convex. We tackle this by constraining $\omega_{kl}^2\leq \gamma_{kl}\leq \omega_{kl}^2$, which can be convexified as
  \begin{equation}\label{quadratic relaxation}
     \omega_{kl}^2\leq \gamma_{kl}\leq 2w_{kl}^{(i)}\omega_{kl}-\left(\omega_{kl}^{(i)}\right)^2.
 \end{equation}
	\subsubsection{Address Constraint \eqref{cons_CRB}}
 
 We introduce a new variable $\vartheta$ and rewrite \eqref{cons_CRB} as
 \begin{subequations}
\begin{empheq}                      [left=\eqref{cons_CRB}\Leftrightarrow\empheqlbrace]{align}
        &\frac{\bar{\sigma}_{\mathtt{s}}}{\mathtt{CRLB}_{\theta_p}^\mathtt{th}} \tau_{p}(\bm{\Psi}) + \abs{ \hat{\tau}_{p} (\bm{\Psi})}^2 \leq \vartheta_p^2,  \forall p &\label{cons_CRB3a}  \\
        &\tau_{p}(\bm{\Psi}) \tilde{\tau}_{p}(\bm{\Psi}) \geq \vartheta_p^2,  \forall p. &\label{cons_CRB3b}
\end{empheq} \label{cons_CRB3}
\end{subequations}
Note that in \eqref{def_tau}--\eqref{def_tau_hat}, $\norm{\bar{\bm{\eta}}_p}^2 = \sum_{p = 1}^L \eta_{kp}$. Thus, $\tau_{p}(\bm{\Psi}), \tilde{\tau}_p(\bm\Psi) $ and $\hat{\tau}_p(\bm\Psi)$ are linear with respect to $\bm\Psi$, and constraint \eqref{cons_CRB3a} can be convexified at iteration $i$ as
\begin{align*}
    \frac{\bar{\sigma}_{\mathtt{s}}}{t} \tau_{p}(\bm{\Psi} ) + \abs{ \hat{\tau}_{p} (\bm{\Psi})}^2 \leq 2 \vartheta_p^{(i)} \vartheta_p - \left(\vartheta_p^{(i)}\right)^2,\  \forall p. \nbthis \label{cons_CRB5}
\end{align*}
Furthermore, \eqref{cons_CRB3b} can be recast into the following second-order constraint (SOC):
\begin{align*}
    \norm{\vartheta_p, \frac{1}{2} \left(\tau_{p}(\bm{\Psi}) - \tilde{\tau}_{p}(\bm{\Psi}) \right)}_2 \leq \frac{1}{2} \left( \tau_{p}(\bm{\Psi}) +\tilde{\tau}_{p}(\bm{\Psi}) \right). \nbthis \label{cons_CRB6}
\end{align*}
As a result, \eqref{problem1} can be solved sequentially by the following problem at iteration $i$:
\begin{subequations}\label{problem3}
    \begin{align*}
        \underset{\substack{\{\bm{\omega}_k\}_{k=1}^K,\bm{\Psi}}}{\textrm{maximize}} \quad &  \sum\nolimits_{k=1}^K f_k(\bm{\omega}_k,\bm\Psi) \nbthis \label{obj_CRB_2} \\
        \textrm{subject to} \quad
        & \eqref{cons_power}, \eqref{quadratic relaxation},\eqref{cons_CRB5}, \eqref{cons_CRB6}, \nbthis
    \end{align*}
\end{subequations}
which is convex and can be solved with standard convex optimization solver such as CVX. The proposed power allocation method is outlined in Algorithm \ref{alg_power_allocation}, which is self-explanatory. 

	\addtolength{\topmargin}{0.01in}
	\setlength{\textfloatsep}{7pt}	
	\begin{algorithm}[h]
		\small
		\textbf{Initialize}: $i = 0$, $\bm{\gamma}^{(0)}$, $\bm{\eta}^{(0)}$,  and $\vw^{(0)}$.\\
		\Repeat{the objective value in \eqref{obj_CRB_2} converges}{
			$i = i+1$.\\
			Update $\{\zeta_k^{(i)}\}$ and $\{\iota_k^{(i)}\}$ by \eqref{eq_zeta} and \eqref{eq_iota}, respectively.\\
			Solve \eqref{problem3} to obtain optimal solution $(\bm{\gamma}^{\star}, \bm{\eta}^{\star}, \vw^{\star})$.\\
		}	
		\caption{Proposed Power Allocation Scheme}
		\label{alg_power_allocation}				
	\end{algorithm}
	
	\begin{figure*}[t]%
		\vspace{-0.5cm}
		\hspace{-0.65cm}%
		\subfigure[Sum rate vs. $L$]{%
			\includegraphics[height=1.95in]{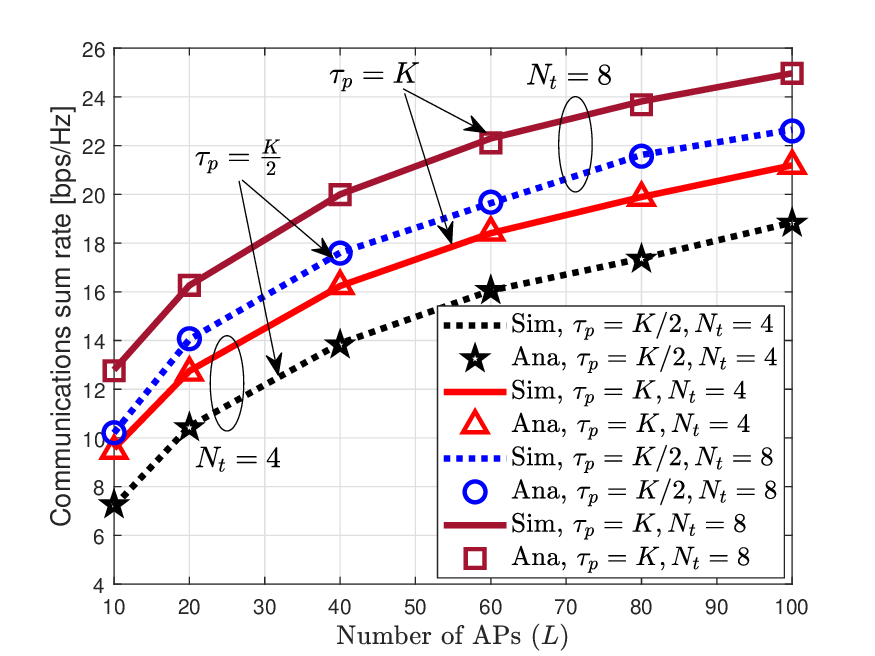}
			\label{fig_rate_L}} 
		\hspace{-0.4cm}
		\subfigure[Convergence]{%
			\includegraphics[height=1.95in]{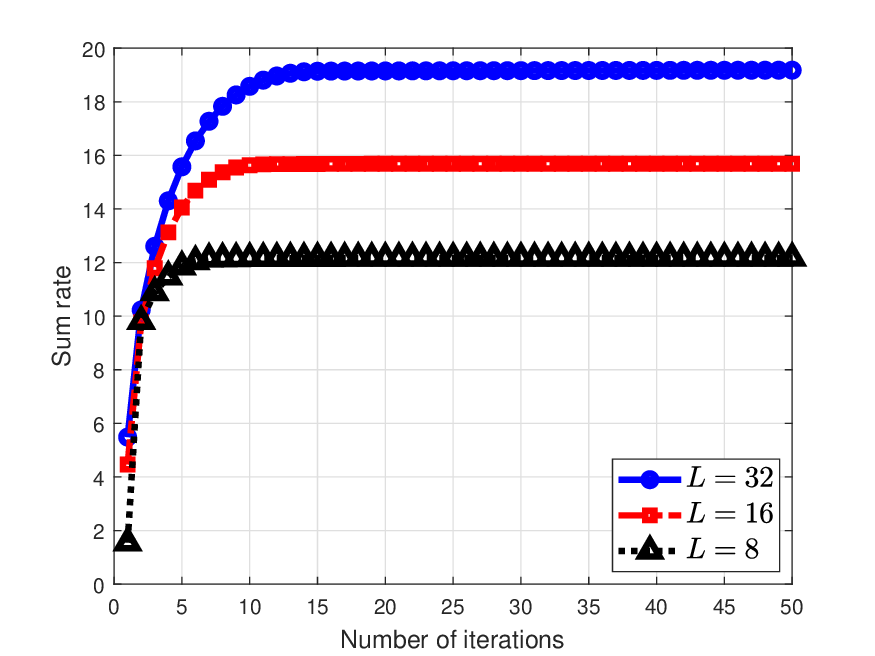}
			\label{fig_conv}}%
		\hspace{-0.55cm}%
		\subfigure[Sum rate vs. CRLB threshold]{%
			\includegraphics[height=1.95in]{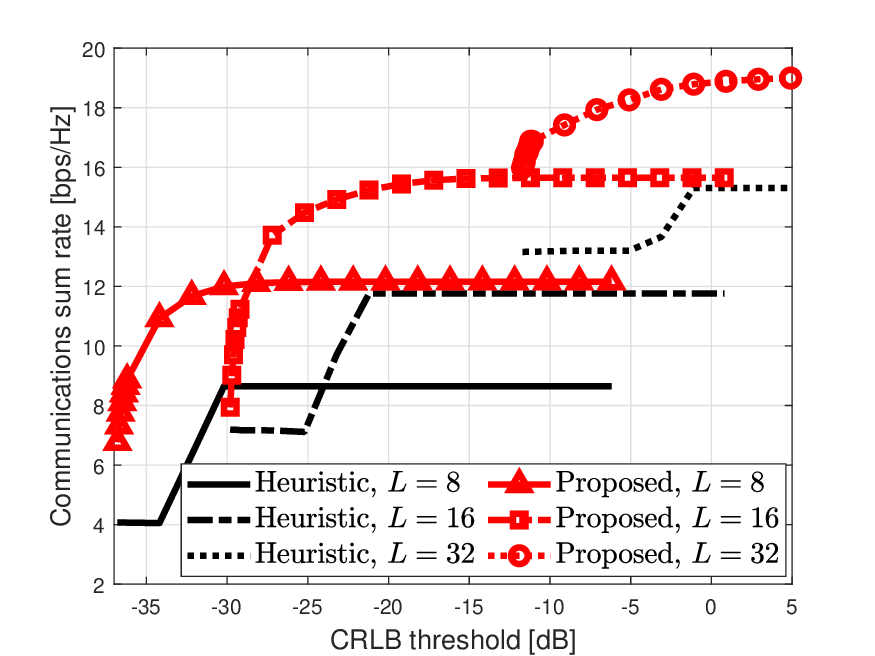}
			\label{fig_rate_vs_CRB0}} 
		\caption{Convergence and performance of the proposed power allocation scheme in the considered CF mMIMO ISAC system with $L=\{8,16,32\}$, $\Nt=\Nr=4$, $K=4$. In Figs.~(a) and (b), SNR $=30$ dB; in Fig.~(c), $\mathtt{CRLB}_{\theta_p}^\mathtt{th} = -5$ dB, $\forall p$.}
		\label{fig_SE}%
	\end{figure*}
	
	\section{Simulation Results}
	 \label{sec_sim}

	 In this section, we provide numerical results to validate the theoretical findings and proposed design. We consider a scenario where the APs and UEs have random locations that are uniformly distributed within an area  of $250 \times 250~\text{m}^2$, and their is no AP-UE pair with distance smaller than $r_{\mathtt{h}} = 100$ m. The large-scale coefficients are computed as $\beta_{k\ell} = z_{k\ell}/(r_{k\ell}/r_{\mathtt{h}})^{\nu}$, where $z_{k\ell}$ is a log-normal random variable with standard deviation $\sigma_{\mathtt{shadow}} = 7$ dB, $r_{k\ell}$ is the distance between the $k$-th UE and AP $\ell$, and $\nu = 3.2$ is the path loss exponent~\cite{ngo2013energy}. For simplicity, we assume $\alpha = \frac{\beta_{\mathtt{s}}}{\sqrt{2}}(1 + j)$, where $\beta_{\mathtt{s}} = 10^{-2}$. We set $\Nt=\Nr=4$, $K=4$,  $T=30$, $\sigmac = \sigmas = -30$ dBm, and the SNR is defined as $\mathrm{SNR} = \frac{\Pt}{\noise}$. 
		
	 In Fig.\ \ref{fig_rate_L}, we show the communications sum rate versus the number of APs, i.e., $L$, for $\Nt = \{4,8\}$ and $\tau_{\mathtt{p}} = \{K/2,K\}$. The results are obtained by both the closed-form expression \eqref{eq_SE_theo} and Monte Carlo simulations for $100$ small-scale and $10$ large-scale channel realizations. To focus on verifying the theoretical findings in Theorem \ref{theo_SE} and Remark \ref{rm_effect_of_sensing}, we assume equal power allocations among communications and sensing for all APs, i.e., we set $\left(\eta_{k \ell},\gamma_{k \ell}\right) = \left( \frac{\Pt}{2\Nt L K},  \frac{\Pt}{2\Nt \sum_{k=1}^K {\xi}_k} \right), \forall k, \ell$. It is clear that the closed-form achievable rates in Theorem \ref{theo_SE} align well with those obtained via Monte Carlo simulations. Furthermore, when $L$ and/or $\Nt$ increases, the sum rate increases significantly, verifying the finding in Remark \ref{rm_effect_of_sensing}. The figure also show the performance loss due to pilot contamination when setting $\tau_{\mathtt{p}} < K$.

	 In Fig.\ \ref{fig_conv} we show the convergence of Algorithm \ref{alg_power_allocation} for $L = \{8, 16, 32\}$. The algorithm is initialized with a heuristic power allocation scheme, in which the available power budget is first shared between communications and sensing with factors $\{\rho, 1-\rho\}$, respectively, and then the power for each subsystem is shared equally among data streams. Specifically, we set $\left(\eta_{k \ell},\gamma_{k \ell}\right) = \left( \frac{\rho \Pt}{\Nt L K},  \frac{(1-\rho)\Pt}{\Nt \sum_{k=1}^K {\xi}_k} \right)$ with $\rho = 0.3$ based on empirical simulation. Then, the power factor for sensing $\eta_{k \ell}$ is increased until AP $\ell$ meets the CRLB constraint. It is observed that Algorithm \ref{alg_power_allocation} converges well after $\{8,10,15\}$ iterations for $L = \{8, 16, 32\}$, respectively. It is reasonable that the convergence becomes slower as $L$ increases because the number of design variables is proportional to $L$. 
	
	 In Fig.\ \ref{fig_rate_vs_CRB0}, we show the communications sum rate versus the CRLB thresholds, i.e., $\mathtt{CRLB}_{\theta_p}^\mathtt{th}$ in the optimization problem \ref{problem1}. For $L = \{8, 16, 32\}$, we set different ranges of $\mathtt{CRLB}_{\theta_p}^\mathtt{th}$ to ensure that the constraint \eqref{cons_CRB} is feasible. We compare the performance of Algorithm \ref{alg_power_allocation} to those obtained by the heuristic power allocation discussed above. First, we see that increasing $L$, i.e., deploying more APs, leads to a narrower feasible space, which is reasonable because \eqref{cons_CRB} applies for all the APs. However, for the same $\mathtt{CRLB}_{\theta_p}^\mathtt{th}$, the larger $L$ always offers better sum rate. This aligns well with Remark \ref{rm_effect_of_sensing}.

     \section{Conclusion}
	 \label{sec_conclusion}
	 We have investigated a multi-static CF mMIMO-ISAC system with MRT beamforming at the APs. By deriving closed-form expressions for the achievable communications rate and the CRLB, we have shown the mutual effects of communications and sensing as well as the advantages of using numerous distributed APs in the system. We then optimized the power allocation among the APs to maximize the communications rate subject to the constraints on the sensing performance and total transmit power budget. Finally, we presented numerical results to validate our theoretical findings and demonstrate the efficiency of the proposed power allocation approach.
	
	\section*{Acknowledgement}
	This work was supported in part by the Research Council of Finland through 6G Flagship under Grant 346208 and through project DIRECTION under Grant 354901, through CHIST-ERA PASSIONATE project (grant number 359817), Business Finland, Keysight, MediaTek, Siemens, Ekahau, and Verkotan via project 6GLearn.
	
	\appendices

\section{Proof of Theorem \ref{theo_SE}}
	\label{appd_SE}
	
	To derive the closed-form SE expressions, in the following, we compute $\abs{\mathtt{DS}_k}^2$, $\meanshort{\abs{\mathtt{BU}_k}^2}$, and $\meanshort{\abs{\mathtt{UI}_{kj}}^2}$. We first have
	\begin{align*}
		\vh_{k \ell}^\H \vf_{i \ell} &= (\hat{\vh}_{k \ell} + \ve_{k \ell})^\H \left(\sqrt{\gamma_{i \ell}} \hat{\vh}_{i \ell} + \sqrt{\eta_{i \ell}} \va\right)\\
		&= \sqrt{\gamma_{i \ell}} \hat{\vh}_{k \ell}^\H \hat{\vh}_{i \ell} + \sqrt{\eta_{i \ell}}  \hat{\vh}_{k \ell}^\H \va\\
		&\hspace{2cm}  + \ve_{k \ell}^\H (\sqrt{\gamma_{i \ell}} \hat{\vh}_{i \ell} + \sqrt{\eta_{i \ell}} \va),\ \forall i,\ell, \nbthis \label{eq_hkfmrti}
	\end{align*}
	where we note that
	\begin{align*}
		\mean{\sqrt{\eta_{i \ell}}  \hat{\vh}_{k \ell}^\H \va + \ve_{k \ell}^\H (\sqrt{\gamma_{i \ell}} \hat{\vh}_{i \ell} + \sqrt{\eta_{i \ell}} \va)} = 0, \nbthis \label{eq_nean_0}
	\end{align*}
	as $\ve_{k\ell}$ and $\hat{\vh}_{i\ell}$ are independent and they both have zero means.
	
	
	\subsubsection{Computation of $\abs{\mathtt{DS}_k}^2$}
	Using \eqref{eq_hkfmrti} and \eqref{eq_nean_0}, we obtain
 \begin{align*}
		\mathtt{DS}_k &= \sum_{\ell =1}^{L} \sqrt{\gamma_{k \ell}} \mean{\normshort{\hat{\vh}_{k \ell}}^2} = \Nt \sum_{\ell =1}^{L}  \xi_{k \ell} \sqrt{\gamma_{k \ell}} = \Nt \bm{\xi}_k^\T \bar{\bm{\gamma}}_k,
	\end{align*}
    where we have used the definitions of $\bm{\xi}_k$ and $\bar{\bm{\gamma}}_k$ in Theorem \ref{theo_SE}. As a result, $\abs{\mathtt{DS}_k}^2$ is given as
    \begin{align*}
		\abs{\mathtt{DS}_k}^2 &= \Nt^2 \left(\bm{\xi}_k^\T \bar{\bm{\gamma}}_k\right)^2. \nbthis \label{eq_abs_DSk2}
	\end{align*}
	
	\subsubsection{Computation of $\meanshort{\abs{\mathtt{BU}_k}^2}$}
	Because the variance of a sum of independent random variables is equal to the sum of the variances, we have
	\begin{align*}
		\mean{\abs{\mathtt{BU}_k}^2} &= \mean{\abs{\sum_{\ell =1}^{L}  \vh_{k \ell}^\H {\vf}_{k \ell} - \mean{ \sum_{\ell =1}^{L} \vh_{k \ell}^\H {\vf}_{k \ell}}}^2}\\
		&= \sum_{\ell =1}^{L} \mean{\abs{ \vh_{k \ell}^\H {\vf}_{k \ell} - \mean{\vh_{k \ell}^\H {\vf}_{k \ell}}}^2}\\
		&= \sum_{\ell =1}^{L}  \underbrace{\mean{\abs{\vh_{k \ell}^\H {\vf}_{k \ell}}^2}}_{\triangleq E_{0k \ell}} - \underbrace{\abs{\mean{\vh_{k \ell}^\H {\vf}_{k \ell}}}^2}_{\triangleq E_{1k\ell}}.\nbthis \label{eq_nean_abs_BUk2}
	\end{align*}
	Here, $E_{0k \ell}$ can be computed as:
	\begin{align*}
		E_{0k \ell} &= \mathbb{E} \left\{ \left|\sqrt{\gamma_{k \ell}} \normshort{\hat{\vh}_{k \ell}}^2 + \sqrt{\eta_{k \ell}}  \hat{\vh}_{k \ell}^\H \va \right. \right.\\
		&\hspace{2cm} \left. \left.  + \ve_{k \ell}^\H (\sqrt{\gamma_{k \ell}} \hat{\vh}_{k \ell} + \sqrt{\eta_{k \ell}} \va)\right|^2 \right\}\\
		&= \gamma_{k \ell} \mean{\normshort{\hat{\vh}_{k \ell}}^4} + \eta_{k \ell} \mean{\absshort{\hat{\vh}_{k \ell}^\H \va}^2} \\
		&\hspace{2cm} + \mean{\absshort{\ve_{k \ell}^\H (\sqrt{\gamma_{k \ell}} \hat{\vh}_{k \ell} + \sqrt{\eta_{k \ell}} \va)}^2}. \nbthis \label{eq_E000}
	\end{align*}
	Note that the elements of $\va$ are deterministic with unit modulus, and $\hat{\vh}_{k \ell} \sim \mathcal{CN}(0, \xi_{k \ell} \mI_{\Nt})$. Thus, we obtain $\mean{\normshort{\hat{\vh}_{k \ell}}^4} = \Nt (\Nt+1) \xi_{k \ell}^2$, $\mean{\absshort{\hat{\vh}_{k \ell}^\H \va}^2} = \mean{\normshort{\hat{\vh}_{k \ell}}^2} = \Nt \xi_{k \ell}$, and $\mean{\absshort{\ve_{k \ell}^\H (\sqrt{\gamma_{k \ell}} \hat{\vh}_{k \ell} \!+\! \sqrt{\eta_{k \ell}} \va)}^2} = \Nt \epsilon_{k \ell}\! \left(\xi_{k \ell} \gamma_{k \ell} \!+\! \eta_{k \ell} \right)$. Hence, $E_{0k\ell}$ in \eqref{eq_nean_abs_BUk2} can be obtained as
	\begin{align*}
		E_{0k\ell} 
		&= \Nt^2 \xi_{k \ell}^2 \gamma_{k \ell} + \Nt \beta_{k \ell} \left( \xi_{k \ell} \gamma_{k \ell} + \eta_{k \ell} \right), \nbthis \label{eq_E0001}
	\end{align*}
	where we have used $\xi_{k \ell} + \epsilon_{k \ell} = \beta_{k \ell}$. Furthermore, it follows \eqref{eq_hkfmrti} and \eqref{eq_nean_0} that
 \begin{align*}
     E_{1k\ell} = \abs{\sqrt{\gamma_{k \ell}} \mean{\norm{\hat{\vh}_{k \ell}}^2}}^2 = \Nt^2 \xi_{k\ell}^2 \gamma_{k \ell}. \nbthis \label{eq_E1}
 \end{align*}
 From \eqref{eq_nean_abs_BUk2}, \eqref{eq_E0001}, and \eqref{eq_E1}, we obtain
	\begin{align*}
		\mean{\abs{\mathtt{BU}_k}^2} &= \Nt \sum_{\ell =1}^{L} \beta_{k \ell}  \left( \xi_{k \ell} \gamma_{k \ell} + \eta_{k \ell} \right)\\
		&= \Nt \left( \hat{\bm{\beta}}_{kk}^\T \bm{\gamma}_k + \bm{\beta}_{k}^\T \bm{\eta}_k \right), \nbthis \label{eq_nean_abs_BUk21}
	\end{align*}
	where $\hat{\bm{\beta}}_{kk}$ and $\bm{\beta}_{k}$ are defined in Theorem \ref{theo_SE}.
	
	\subsubsection{Computation of $\meanshort{\abs{\mathtt{UI}_{kj}}^2}$}
	We compute this term as
	\begin{align*}
		&\mean{\abs{\mathtt{UI}_{kj}}^2} = \mean{\abs{\sum\nolimits_{\ell =1}^{L} \vh_{k \ell}^\H {\vf}_{j\ell}}^2}\\
		&= \mean{\abs{\sum\nolimits_{\ell =1}^{L} \sqrt{\gamma_{j\ell}} \hat{\vh}_{k \ell}^\H \hat{\vh}_{j\ell}}^2} + \mean{\abs{\sum\nolimits_{\ell =1}^{L} \sqrt{\eta_{j\ell}} \hat{\vh}_{k \ell}^\H \va_{\ell}}^2} \\
		&\qquad + \mean{\abs{\sum\nolimits_{\ell =1}^{L} \ve_{k \ell}^\H (\sqrt{\gamma_{j\ell}} \hat{\vh}_{j\ell} + \sqrt{\eta_{j\ell}} \va_{\ell})}^2}\\
		&= \sum\nolimits_{\ell =1}^{L} \gamma_{j\ell} \mean{\abs{\hat{\vh}_{k \ell}^\H \hat{\vh}_{j\ell}}^2} + \sum\nolimits_{\ell =1}^{L} \eta_{j\ell} \mean{\abs{ \hat{\vh}_{k \ell}^\H \va_{\ell}}^2} \\
		&\qquad + \mean{\abs{\sum\nolimits_{\ell =1}^{L} \ve_{k \ell}^\H (\sqrt{\gamma_{j\ell}} \hat{\vh}_{j\ell} + \sqrt{\eta_{j\ell}} \va_{\ell})}^2}\\
		&= \Nt \sum\nolimits_{\ell =1}^{L} \left(\gamma_{j\ell}  \xi_{k \ell} \xi_{j\ell} + \eta_{j\ell} \xi_{k \ell} + \epsilon_{k \ell} (\gamma_{j\ell} \xi_{j\ell} + \eta_{j\ell})\right) \\
		&= \Nt \left( \hat{\bm{\beta}}_{kj}^\T \bm{\gamma}_j + \bm{\beta}_{k}^\T \bm{\eta}_j  \right). \nbthis \label{eq_nean_abs_UIk2}
	\end{align*}
	From \eqref{eq_abs_DSk2}, \eqref{eq_nean_abs_BUk21}, and \eqref{eq_nean_abs_UIk2}, we obtain \eqref{eq_SE_theo}.
	
	\section{Proof of Theorem \ref{theo_CRB}}
	\label{appd_CRB}
 
		Recalling that $\tilde{\vx}^{\mathtt{s}} _{p} = \mathtt{vec}(\mX^{\mathtt{s}} _{p})$, with $\mX^{\mathtt{s}} _{p}$ given in \eqref{def_X}, it follows that $\frac{\partial \tilde{\vx}^{\mathtt{s}}_{p}}{\partial \theta_{p}} = \mathtt{vec}(\sum_{\ell=1}^{L} \alpha_{\ell p} \dot{\mB}_{\ell p} \mX)
		=\mathtt{vec}(\dot{\tilde{\mB}}_{p} \mX)$ 
		where $\dot{\tilde{\mB}}_{p} \triangleq \sum_{\ell=1}^{L} \alpha_{\ell p} \dot{\mB}_{\ell p}$. Based on \eqref{def_Blplp} with the note that $\bar{\va}_{p}$ is independent of $\theta_{p}$, we have $\dot{\mB}_{\ell p} = \sum_{\ell=1}^{L} (\dot{\vb}_{p}  \bar{\va}_{\ell}^\H + \delta_{\ell p} \vb_{p}  \dot{\bar{\va}}_{\ell}^\H ) \mD_{p}$, where $\dot{\vb}_{p} \! =  \! \frac{\partial \vb_p}{\partial \theta_{p}} \!  = \!  \left[-j\pi\cos(\theta_{p}),\ldots,-j\pi\Nr\cos(\theta_{p})\right]^\T \odot \vb_p$ and $\dot{\bar{\va}}_{\ell}^\H  =  \! \frac{\partial \bar{\va}_{\ell}^\H}{\partial \theta_{p}} = \left[\bm{0}, \ldots, \left[-j\pi\cos(\phi_{\ell}),\ldots,-j\pi\Nt\cos(\phi_{\ell})\right]^\T \odot \va_{\ell}, \ldots, \bm{0}\right]$,~with $\odot$ denoting the Hadamard product of two vectors.
		Follow similar steps as in \cite{song2023intelligent}, the (block) entries of $\mJ_{\bm{\omega}_{p}}$ can be obtained as:
		\begin{align*}
			J_{\theta \theta}^{(p)} 
			&= \bar{\sigma}_{\mathtt{s}}^{-1}  \tr{ \dot{\tilde{\mB}}_{p} \mR_x \dot{\tilde{\mB}}_{p}^\H }, \nbthis \label{eq_J11_result} \\
			\mJ_{\theta \tilde{\bm{\alpha}}}^{(p)} 
			&= \bar{\sigma}_{\mathtt{s}}^{-1}  \left\{ [1,j]  \tr{\mB_{p p}  \mR_x \dot{\tilde{\mB}}_{p}^\H }\right\}, \nbthis \label{eq_J12_result}\\
			\mJ_{\tilde{\bm{\alpha}} \tilde{\bm{\alpha}}}^{(p)} 
			&= \bar{\sigma}_{\mathtt{s}}^{-1}  \tr{ \mB_{p p} \mR_x \mB_{p p}^\H } \mI_2, \nbthis \label{eq_J22_result}
		\end{align*}
		where $\bar{\sigma}_{\mathtt{s}} \triangleq \frac{\sigmas}{2T}$ and $\mR_x =  \frac{1}{T} \mean{\mX \mX^\H} \in \setC^{L\Nt \times L\Nt}$. From \eqref{def_X}, it is observed that $\mR_x$ include $L \times L$ sub-matrices, and the $(\ell,p)$-th one is computed as $[\mR_x]_{\ell p} = \frac{1}{T}\mean{\mX_{\ell} \mX_{p}^\H} = \mean{\mF_{\ell} \mS \mS^\H \mF_{p}^\H} = \mean{\mF_{\ell} \mF_{p}^\H}$, where $\mF_{\ell}$ is given in \eqref{eq_F} and $\mS \mS^\H = T \mI_K$. Note that $\mathbb{E}\{\hat{\mH}_{\ell}\} = \bm{0}$. After some manipulation, we obtain
		\begin{align*}
			[\mR_x]_{\ell p} 
			&= \sum_{k=1}^K \sqrt{\gamma_{k\ell }} \sqrt{\gamma_{k p}} \mean{\hat{\vh}_{k \ell} \hat{\vh}_{k p}^\H} + \bar{\bm{\eta}}_\ell^\T \bar{\bm{\eta}}_{p}^* \va_\ell  \va_{p}^\H \\
			&= \delta_{\ell p} \bm{\xi}_{p}^\T \bm{\gamma}_{p} \mI_{\Nt} + \bar{\bm{\eta}}_{p}^\T \bar{\bm{\eta}}_\ell \tilde{\mA}_{\ell p}. \nbthis \label{eq_Rx}
		\end{align*}
		where $\bm{\gamma}_\ell$ and $\tilde{\mA}_{\ell p}$ are defined in Theorem \ref{theo_CRB}. The last equality follows the fact that $\hat{\vh}_{k \ell}$ and $ \hat{\vh}_{k p}$ are both zero-mean and independent vectors for $\ell \neq p$, while $\mean{\hat{\vh}_{k p} \hat{\vh}_{k p}^\H} = \xi_{k \ell} \mI_{\Nt}$. As a result, we can derive the CRLB for $\theta_{p}$ as
		\begin{align*}
			&\mathtt{CRLB}_{p}^{\theta}(\bm{\Psi}) = \left[\mJ_{\bm{\omega}_{p}}^{-1}\right]_{11} \! 
			= \frac{\bar{\sigma}_{\mathtt{s}} \tau_{p}(\bm{\Psi})}{ \tilde{\tau}_{p}(\bm{\Psi}) \tau_{p} (\bm{\Psi}) - \abs{ \hat{\tau}_{p} (\bm{\Psi})}^2}
		\end{align*}
		where $\tau_{p}(\bm{\Psi}) \triangleq \trshort{ \mB_{p p} \mR_x \mB_{p p}^\H }$,  $\tilde{\tau}_{p}(\bm{\Psi}) \triangleq \trshort{ \dot{\tilde{\mB}}_{p} \mR_x \dot{\tilde{\mB}}_{p}^\H }$, and $\hat{\tau}_{p}(\bm{\Psi}) \triangleq \trshort{\mB_{p p}  \mR_x \dot{\tilde{\mB}}_{p}^\H }$. 
		
		To further expose the role of the power factors in $\mathtt{CRLB}_{p}^{\theta}(\bm{\Psi})$, we rewrite $\tau_{p}(\bm{\Psi}) = \sum_{\ell=1}^{L} \sum_{q=1}^{L} \trshort{ [\mR_x]_{\ell q} \left[ \mB_{p p}^\H \mB_{p p} \right]_{q \ell} }$. From \eqref{def_Blplp} and \eqref{def_Dl}, it is observed that $\left[ \mB_{p p}^\H \mB_{p p} \right]_{q \ell} = \bm{0}_{\Nt}, \forall \ell, q \neq p$. Thus, based on \eqref{eq_Rx}, we obtain
		\begin{align*}
			\tau_{p}(\bm{\Psi}) &= \tr{ [\mR_x]_{pp} \left[ \mB_{p p}^\H \mB_{p p} \right]_{pp} } \\
			&= \bm{\xi}_{p}^\T \bm{\gamma}_{p} c_{p,1} + \norm{\bar{\bm{\eta}}_{p}}^2 c_{p,2}, \nbthis \label{def_tau_1}
		\end{align*}
		where $c_{p,1}$ and $c_{p,2}$ are defined in Theorem \ref{theo_CRB}. Similarly, we can obtain \eqref{def_tau_tilde} and \eqref{def_tau_hat}, and the proof is complete.
	
	\bibliographystyle{IEEEtran}
	\bibliography{IEEEabrv,Bibliography,CWC_ISAC}

\begin{thebibliography}{10}
\providecommand{\url}[1]{#1}
\csname url@samestyle\endcsname
\providecommand{\newblock}{\relax}
\providecommand{\bibinfo}[2]{#2}
\providecommand{\BIBentrySTDinterwordspacing}{\spaceskip=0pt\relax}
\providecommand{\BIBentryALTinterwordstretchfactor}{4}
\providecommand{\BIBentryALTinterwordspacing}{\spaceskip=\fontdimen2\font plus
\BIBentryALTinterwordstretchfactor\fontdimen3\font minus
  \fontdimen4\font\relax}
\providecommand{\BIBforeignlanguage}[2]{{%
\expandafter\ifx\csname l@#1\endcsname\relax
\typeout{** WARNING: IEEEtran.bst: No hyphenation pattern has been}%
\typeout{** loaded for the language `#1'. Using the pattern for}%
\typeout{** the default language instead.}%
\else
\language=\csname l@#1\endcsname
\fi
#2}}
\providecommand{\BIBdecl}{\relax}
\BIBdecl

\bibitem{liu2018mu}
F.~Liu, C.~Masouros, A.~Li, H.~Sun, and L.~Hanzo, ``{MU-MIMO communications
  with MIMO radar: From co-existence to joint transmission},'' \emph{{IEEE}
  Trans. Wireless Commun.}, vol.~17, no.~4, pp. 2755--2770, 2018.

\bibitem{liu2020joint}
X.~Liu, T.~Huang, N.~Shlezinger, Y.~Liu, J.~Zhou, and Y.~C. Eldar, ``Joint
  transmit beamforming for multiuser {MIMO} communications and {MIMO} radar,''
  \emph{{IEEE} Trans. Signal Process.}, vol.~68, pp. 3929--3944, 2020.

\bibitem{johnston2022mimo}
J.~Johnston, L.~Venturino, E.~Grossi, M.~Lops, and X.~Wang, ``{MIMO OFDM
  dual-function radar-communication under error rate and beampattern
  constraints},'' \emph{{IEEE} J. Sel. Areas Commun.}, vol.~40, no.~6, pp.
  1951--1964, 2022.

\bibitem{liu2022joint}
F.~Liu, Y.-F. Liu, C.~Masouros, A.~Li, and Y.~C. Eldar, ``A joint
  radar-communication precoding design based on {Cram{\'e}r-Rao} bound
  optimization,'' in \emph{IEEE Radar Conf.}, 2022.

\bibitem{nguyen2023joint}
N.~T. Nguyen, N.~Shlezinger, K.-H. Ngo, V.-D. Nguyen, and M.~Juntti, ``Joint
  communications and sensing design for multi-carrier {MIMO} systems,'' in
  \emph{Proc. IEEE Works. on Statistical Signal Processing}, 2023, pp.
  110--114.

\bibitem{krishnananthalingam2024constant}
P.~Krishnananthalingam, N.~T. Nguyen, and M.~Juntti, ``Constant modulus
  waveform design for wideband multicarrier joint communications and sensing
  via deep unfolding,'' in \emph{Proc. IEEE Wireless Commun. and Networking
  Conf.}, 2024.

\bibitem{nguyen2024massive}
N.~T. Nguyen, V.-D. Nguyen, H.~V. Nguyen, H.~Q. Ngo, A.~Swindlehurst, and
  M.~Juntti, ``Massive {MIMO} joint communications and sensing with {MRT}
  beamforming,'' in \emph{IEEE Radar Conf.}, 2024.

\bibitem{hatami2024waveform}
M.~Hatami, N.~Nguyen, and M.~Juntti, ``Waveform design for multi-carrier
  multiuser {MIMO} joint communications and sensing,'' in \emph{Proc. IEEE
  Works. on Sign. Proc. Adv. in Wirel. Comms.}\hskip 1em plus 0.5em minus
  0.4em\relax IEEE, 2024, pp. 346--350.

\bibitem{nguyen2023multiuser}
N.~T. Nguyen, N.~Shlezinger, Y.~C. Eldar, and M.~Juntti, ``Multiuser {MIMO}
  wideband joint communications and sensing system with subcarrier
  allocation,'' \emph{{IEEE} Trans. Signal Process.}, vol.~71, pp. 2997--3013,
  2023.

\bibitem{nguyen2024joint}
N.~T. Nguyen, L.~V. Nguyen, N.~Shlezinger, Y.~C. Eldar, A.~L. Swindlehurst, and
  M.~Juntti, ``Joint communications and sensing hybrid beamforming design via
  deep unfolding,'' \emph{{IEEE} J. Sel. Topics Signal Process.}, 2024.

\bibitem{elfiatoure2024multiple}
M.~Elfiatoure, M.~Mohammadi, H.~Q. Ngo, and M.~Matthaiou, ``Multiple-target
  detection in cell-free massive {MIMO}-assisted {ISAC},'' \emph{arXiv preprint
  arXiv:2404.17263}, 2024.

\bibitem{rivetti2024secure}
S.~Rivetti, E.~Bj{\"o}rnson, and M.~Skoglund, ``Secure spatial signal design
  for {ISAC} in a cell-free {MIMO} network,'' in \emph{Proc. IEEE Wireless
  Commun. and Networking Conf.}, 2024.

\bibitem{behdad2024multi}
Z.~Behdad, {\"O}.~T. Demir, K.~W. Sung, E.~Bj{\"o}rnson, and C.~Cavdar,
  ``Multi-static target detection and power allocation for integrated sensing
  and communication in cell-free massive {MIMO},'' \emph{{IEEE} Trans. Wireless
  Commun.}, 2024.

\bibitem{demirhan2023cell}
U.~Demirhan and A.~Alkhateeb, ``Cell-free {ISAC MIMO} systems: {J}oint sensing
  and communication beamforming,'' \emph{arXiv preprint arXiv:2301.11328},
  2023.

\bibitem{demirhan2023cell_asilomar}
------, ``Cell-free joint sensing and communication {MIMO}: {A} max-min fair
  beamforming approach,'' in \emph{Proc. Annual Asilomar Conf. Signals, Syst.,
  Comp.}, 2023, pp. 381--386.

\bibitem{mao2024communication}
W.~Mao, Y.~Lu, C.-Y. Chi, B.~Ai, Z.~Zhong, and Z.~Ding, ``Communication-sensing
  region for cell-free massive {MIMO ISAC} systems,'' \emph{{IEEE} Trans.
  Wireless Commun.}, 2024.

\bibitem{demirhan2024learning}
U.~Demirhan and A.~Alkhateeb, ``Learning beamforming in cell-free massive {MIMO
  ISAC} systems,'' in \emph{Proc. IEEE Works. on Sign. Proc. Adv. in Wirel.
  Comms.}, 2024, pp. 326--330.

\bibitem{adhikary2024holographic}
A.~Adhikary, A.~D. Raha, Y.~Qiao, W.~Saad, Z.~Han, and C.~S. Hong,
  ``Holographic {MIMO} with integrated sensing and communication for
  energy-efficient cell-free {6G} networks,'' \emph{{IEEE} Internet Things J.},
  2024.

\bibitem{zeng2024multi}
F.~Zeng, R.~Liu, X.~Sun, J.~Yu, J.~Li, P.~Zhu, D.~Wang, and X.~You,
  ``Multi-static {ISAC} based on network-assisted full-duplex cell-free
  networks: {P}erformance analysis and duplex mode optimization,'' \emph{arXiv
  preprint arXiv:2406.08268}, 2024.

\bibitem{liu2021cramer}
F.~Liu, Y.-F. Liu, A.~Li, C.~Masouros, and Y.~C. Eldar, ``{Cram{\'e}r-Rao bound
  optimization for joint radar-communication beamforming},'' \emph{{IEEE}
  Trans. Signal Process.}, vol.~70, pp. 240--253, 2021.

\bibitem{mollen2016uplink}
C.~Mollen, J.~Choi, E.~G. Larsson, and R.~W. Heath, ``{Uplink performance of
  wideband massive MIMO with one-bit ADCs},'' \emph{{IEEE} Trans. Wireless
  Commun.}, vol.~16, no.~1, pp. 87--100, 2016.

\bibitem{ngo2017total}
H.~Q. Ngo, L.-N. Tran, T.~Q. Duong, M.~Matthaiou, and E.~G. Larsson, ``On the
  total energy efficiency of cell-free massive {MIMO},'' \emph{{IEEE} Trans.
  Green Commun. Network.}, vol.~2, no.~1, pp. 25--39, 2017.

\bibitem{song2023intelligent}
X.~Song, J.~Xu, F.~Liu, T.~X. Han, and Y.~C. Eldar, ``Intelligent reflecting
  surface enabled sensing: {Cram{\'e}r-Rao} bound optimization,'' \emph{{IEEE}
  Trans. Signal Process.}, vol.~71, pp. 2011--2026, 2023.

\bibitem{bekkerman2006target}
I.~Bekkerman and J.~Tabrikian, ``Target detection and localization using {MIMO}
  radars and sonars,'' \emph{{IEEE} Trans. Signal Process.}, vol.~54, no.~10,
  pp. 3873--3883, 2006.

\bibitem{fang2024beamforming}
T.~Fang, N.~T. Nguyen, and M.~Juntti, ``Beamforming design for max-min fairness
  performance balancing in {ISAC} systems,'' in \emph{Proc. IEEE Works. on
  Sign. Proc. Adv. in Wirel. Comms.}, 2024, pp. 336--340.

\bibitem{ngo2013energy}
H.~Q. Ngo, E.~G. Larsson, and T.~L. Marzetta, ``Energy and spectral efficiency
  of very large multiuser {MIMO} systems,'' \emph{{IEEE} Trans. Commun.},
  vol.~61, no.~4, pp. 1436--1449, 2013.

\end{thebibliography}
\end{document}